\documentclass[11pt,a4]{article}
\usepackage{color}
\usepackage{epsfig}
\usepackage{latexsym}
\usepackage{graphicx} 
\usepackage{subfigure}
\usepackage{array}
\usepackage{amsfonts}

\newtheorem{theorem}{Theorem}

\newtheorem{lemma}[theorem]{Lemma}

\newcommand{\cm}[1]{}

\newenvironment{proof}{\noindent{\bf Proof:}}{
\hspace*{\fill} $\Box$ \vskip \belowdisplayskip}

\begin{document}
\author{Glencora Borradaile \\Oregon State University \and Claire Mathieu\\Brown University 
 \and Theresa Migler \\Oregon State University}
\title{Lower bounds for testing digraph connectivity with one-pass streaming algorithms}
\maketitle

\begin{abstract}
  In this note, we show that three graph properties - strong
  connectivity, acyclicity, and reachability from a vertex $s$ to all
  vertices - each require a working memory of $\Omega (\epsilon m)$ on
  a graph with $m$ edges to be determined correctly with probability
  greater than $(1+\epsilon)/2$.
\end{abstract}

In the streaming model of computation, the input is given as a
sequence, or {\em stream}, of elements. There is no random access to
the elements; the sequence must be scanned in order. The goal is to
process the stream using a small amount of working memory. For an
overview see \cite{muthukrishnan2005}. There has been much research
devoted to the study of streaming algorithms, most notably the
G\"{o}del-prize winning work of Alon, Matias, and Szegedy \cite{ams1996}.


For undirected graph problems, there are many lower bounds in the edge
streaming model. Henzinger, Raghaven, and Rajagopalan presented a
deterministic lower bound of $\Omega({n})$ for the working memory
required for the following undirected graph problems: computing the
connected components, vertex-connected components, and testing graph
planarity of $n$-vertex graphs \cite{hrr1998}.  Feigenbaum, Kannan,
and Zhang show that any exact, deterministic algorithm for computing
the diameter of an undirected graph in the Euclidean plane requires
$\Omega (n)$ bits of working memory \cite{fkz2004}.  Zelke shows that
any algorithm that is able to find a minimum cut of an undirected
graph requires $\Omega(m)$ bits of working memory, this remains true
even if randomization is allowed \cite{zelke2011}.

For directed graphs problems, the ones most likely to come up in
analyses of the internet, much less is known. Henzinger et al. showed
that for any $0<\epsilon <1$, estimating the size of the transitive
closure of a DAG with relative expected error $\epsilon$ requires
$\Omega (m)$ bits of working memory \cite{hrr1998}. Feigenbaum et
al.~\cite{fkmsz2005b} showed that testing reachability from a given
vertex $s$ to another given vertex $t$ requires $\Omega(m)$ bits of
space, Guruswami and Onak~\cite{go2012} showed that even with $p$
passes, the problem requires
$\Omega(n^{1+1/(2(p+1))}/p^{20}\log^{3/2}n)$ bits of space to be
solvable with probability at least 9/10.

As for upper bounds in undirected graphs, there are one-pass
algorithms for connected components, $k$-edge and $k$-vertex
connectivity ($k\leq 3$), and planarity testing that use $O(n\log n)$
bits of working memory \cite{hrr1998}. There is an algorithm that
approximates the diameter within $1+\epsilon$ using $O({1\over
  \epsilon})$ bits \cite{fkz2004}. For upper bounds in directed
graphs, there is an algorithm that computes the exact size of the
transitive closure using $O(m\log n)$ bits of working
memory\cite{hrr1998}.

In this short note, we consider three basic connectivity questions in
directed graphs: determining if a graph is strongly connected,
determining if a graph is acyclic, and determining if a vertex $s$
reaches all other vertices.  A directed graph $G=(V,E)$ is said to be
{\em strongly connected} if for every pair of vertices $u,v\in V$
there is a path from $u$ to $v$ and a path from $v$ to $u$. A directed
graph $G=(V,E)$ is said to be {\em acyclic} if $G$ contains no
cycles. We say that a vertex $s$ {\em reaches} a vertex $v$ if there
is a directed path from $s$ to $v$.

We show that, even with randomization, these graph properties each require
$\Omega (m)$ bits of working memory to be decided with probability greater
than $(1+\epsilon)/2$ by a one-pass streaming algorithm on $n$
vertices and $m$ edges. For these lower bounds we will use simple
reductions from the index problem (or the bit-vector problem) in communication complexity:

\begin{quote}
  Alice has a bit-vector $x$ of length $m$. Bob has an index $i\in\{
  1,2,\ldots ,m\}$ and wishes to know the $i$th bit of $x$. The only
  communication allowed is from Alice to Bob.
\end{quote}

The following is a rewording of Theorem~2 from Ablayev~\cite{ablayev1996}.

\begin{theorem} 
\label{thm:communication}
  For Bob to correctly determine $x_i$ with probability better than
  $(1+\epsilon)\over 2$, $\Omega (\epsilon m)$ bits of
  communication are required.
\end{theorem}

We will now state and prove our main Lemma:

\begin{lemma}\label{thm:prob}
  Any algorithm that correctly determines the following
  graph properties with probability better than $(1+\epsilon)\over
  2$ requires $\Omega (\epsilon m)$ bits of working memory:

  \noindent acyclicity, strong connectivity, and
  reachability of all from $s$.

\end{lemma}

\begin{proof}
  We reduce from the index problem and use
  Theorem~\ref{thm:communication}.  Let $x$ denote the $m$-bit vector
  owned by Alice.  We define the stream using two sets of edges $E_1$,
  $E_2$. The edge stream first has the edges of $E_1$ in arbitrary
  order, followed by the edges of $E_2$, also in arbitrary order.  The
  set $E_1$ is entirely determined by the $m$-bit vector $x$ owned by
  Alice, and the set $E_2$ is entirely determined by the index $i$
  owned by Bob. The graph defined by $E_1\cup E_2$ has $\Omega (\sqrt
  m)$ vertices, and $E_1$ has $O(m)$ edges.  To solve the index
  problem, Alice constructs $E_1$ and simulates the streaming
  algorithm up to the point when $E_1$ has arrived, then sends to Bob
  the current state of the memory. Upon reception of the message, Bob
  constructs $E_2$ and continues the simulation up to the point when
  $E_2$ has finished arriving. Bob's final decision is then determined
  by the outcome of the streaming algorithm.  Thus, the lemma will
  be proved.

  \paragraph{Acyclicity}

  Let $n= \lceil \sqrt m \rceil$ and let $V=L\cup R$, where $L$ and $R$ both have size $n$ and have vertices labeled $0$ through $n-1$.
  $E_1$ is the bipartite graph that has an edge from vertex $j\in L$ to vertex $k\in R$ iff $x$ has a 1 in position $jn+k$. 
 $E_2$ consists of a single edge determined by Bob's bit $i$: let $k=
 i ~\rm{ mod }~ n$ and  $j={{i-k}\over {n}}$. Then $E_2$
  consists of the edge from vertex $k\in R$ to vertex $i\in L$. See
  figures~\ref{fig:acyclic_e1_concrete} and
  \ref{fig:acyclic_e2_concrete} for an illustration of an example $E_1$ and $E_2$.  

Observe that $E_1\cup E_2$ is acyclic iff $x_i=0$, thus the reduction
is complete. 

\begin{figure}[ht]
    \centering

  \subfigure[An example of $E_1$ corresponding to the bit-vector {001011010}.]{
    \includegraphics[scale=.60]{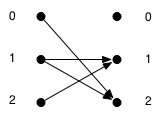}
    \label{fig:acyclic_e1_concrete}
}
  \subfigure[An example of $E_2$ corresponding to $B$'s index being 5,
  $j=1$ and $k=2$.]{
    \includegraphics[scale=.60]{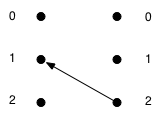}
    \label{fig:acyclic_e2_concrete}
}

\end{figure}

  \paragraph{Strong connectivity} 

  The construction for $E_1$ is the same as in the acyclic case. $E_2$
  consists of $4n-2$ edges determined by Bob's bit $i$: let $k= i
  ~\rm{ mod }~ n$ and $j={{i-k}\over {n}}$. Then $E_2$ consists of all
  edges from $k$ to $V-\{ k\}$, and from $V-\{ j\}$ to $j$. See
  figures~\ref{fig:sc_e1_concrete} and \ref{fig:sc_e2_concrete} for an
  illustration of an example $E_1$ and $E_2$.

We claim that $G$ is strongly connected iff $x_i=1$. Indeed, if $G$ is
strongly connected, there must be a path from $j$ to $k$. The only
edges leaving $j$ are to vertices in $R$ and the only edges entering
$k$ are from vertices in $L$. And the only edges extending from $R$ to $L$
are either entering $j$ or leaving $k$. Thus, the only possible path
from $j$ to $k$ is the single edge from $j$ to $k$ which is present
only when $x_i =1$. Now suppose that $x_i =1$. $k$ can certainly reach
every vertex and every vertex can reach $j$. Since the edge from $j$
to $k$ is present, we know that every vertex can reach $k$ and $k$ can
reach every vertex. Therefore, $G$ is strongly connected. Thus the reduction is complete.

\begin{figure}[ht]
    \centering

  \subfigure[An example of $E_1$ corresponding to the bit-vector {001011010}.]{
    \includegraphics[scale=.60]{acyclic_e1_concrete.png}
    \label{fig:sc_e1_concrete}
}
  \subfigure[An example of $E_2$ corresponding to $B$'s index being 5,
  $j=1$ and $k=2$.]{
    \includegraphics[scale=.60]{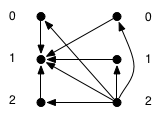}
    \label{fig:sc_e2_concrete}
}

\end{figure}






\paragraph{Reachability from $\mathbf{s}$}
$E_1$ is as above with additional vertex $s$ with in and out
degree $0$.

$E_2$ consists of $2n -1$ edges determined by Bob's bit $i$: let $k= i
~\rm{ mod }~ n$ and $j={{i-k}\over {n}}$. Then $E_2$ consists of one
edge from $s$ to $j\in L$, $n-1$ edges from $j\in L$ to $R-\{ k\}$,
and $n-1$ edges from $k\in R$ to $L-\{ j\}$. See
figures~\ref{fig:s-reach_e1_concrete} and
\ref{fig:s-reach_e2_concrete} for an illustration of an example $E_1$
and $E_2$.

  We claim that in $G$ $s$ reaches everything iff $x_i=1$. Indeed, if
  $s$ can reach all vertices in $G$, and the only edge from $s$ is to
  $j$, $j$ must be able to reach all vertices in $G -\{s\}$. In
  particular $j$ must reach $k$. The only edges extending from $R$ to
  $L$ are from $k$, so the only way for $j$ to reach $k$ is by the
  edge from $j$ to $k$ which is present only when $x_i =1$. Now
  suppose $x_i =1$. We know $s$ reaches $j$ and therefore all of $R$,
  including $k$, and $k$ reaches all of $L-\{j\}$. Therefore, $s$ reaches
  all vertices of $G$. Thus the reduction is complete.

\begin{figure}[ht]
    \centering

  \subfigure[An example of $E_1$ corresponding to the bit-vector {001011010}.]{
    \includegraphics[scale=.60]{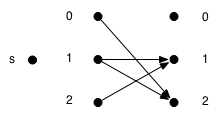}
    \label{fig:s-reach_e1_concrete}
}
  \subfigure[An example of $E_2$ corresponding to $B$'s index being 5,
  $j=1$ and $k=2$.]{
    \includegraphics[scale=.60]{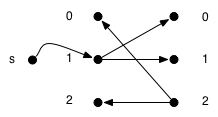}
    \label{fig:s-reach_e2_concrete}
}

\end{figure}

\end{proof}

\bibliographystyle{plain}
\bibliography{papers_read_at_brown}

\end{document}